\documentclass[a4paper,USenglish]{lipics}

\usepackage{epic,gastex}
\usepackage{array}
\usepackage{cite}
\usepackage{enumerate}

\usepackage{tikz}
\usetikzlibrary{arrows,automata, positioning}
\definecolor{light-gray}{gray}{0.8}

\DeclareSymbolFont{rsfscript}{OMS}{rsfs}{m}{n}
\DeclareSymbolFontAlphabet{\mathrsfs}{rsfscript}

\DeclareMathOperator{\rt}{rt}
\DeclareMathOperator{\car}{car}

\theoremstyle{definition}
\newtheorem{problem}[theorem]{Problem}

\usepackage{microtype}

\title{Primitive sets of nonnegative matrices and synchronizing automata}

\author[1]{Bal\'azs Gerencs\'er}
\author[2,3]{Vladimir V. Gusev}
\author[2]{Rapha\"el M. Jungers}
\affil[1]{Alfréd Rényi Institute of Mathematics, Hungarian Academy of Sciences\\
  Budapest, Hungary\\
  \texttt{gerencser.balazs@renyi.mta.hu}}
\affil[2]{ICTEAM Institute, Universit{\'e} catholique de Louvain\\
  Louvain-la-Neuve, Belgium\\
  \texttt{\{vladimir.gusev, raphael.jungers\}@uclouvain.be}
  }
\affil[3]{Institute of Mathematics and Computer Science, Ural Federal University\\
  Ekaterinburg, Russia\\
  }

\authorrunning{B.\,Gerencs\'er, V.\,V.\,Gusev, R.\,Jungers} 
\Copyright{Bal\'azs Gerencs\'er, Vladimir V.\,Gusev, Raphael Jungers}
\subjclass{F.1.1 Models of Computation, G.2.1 Combinatorics}
\keywords{Nonnegative matrices, primitive sets of matrices, the exponent of a matrix set, carefully synchronizing automata, the \v{C}ern\'{y} conjecture}

\serieslogo{}\volumeinfo  {Billy Editor and Bill Editors}  {2}  {Conference title on which this volume is based on}  {1}  {1}  {1}\EventShortName{}
\DOI{10.4230/LIPIcs.xxx.yyy.p}
\begin{document}

\maketitle

\begin{abstract}

A set of nonnegative matrices $\mathcal{M}=\{M_1, M_2, \ldots, M_k\}$ is called \emph{primitive} if there exist indices $i_1, i_2, \ldots, i_m$ such that $M_{i_1} M_{i_2} \ldots M_{i_m}$ is positive (i.e. has all its entries $>0$). The length of the shortest such product is called the \emph{exponent} of $\mathcal{M}$.
The concept of primitive sets of matrices comes up in a number of problems within control theory, non-homogeneous Markov chains, automata theory etc. Recently, connections between synchronizing automata and primitive sets of matrices were established. In the present paper, we significantly strengthen these links by providing equivalence results, both in terms of combinatorial characterization, and computational aspects.

We study the maximal exponent among all primitive sets of $n \times n$ matrices, which we denote by $\exp(n)$. 
We prove that $\lim_{n\rightarrow\infty} \tfrac{\log \exp(n)}{n} = \tfrac{\log 3}{3}$, 
and moreover, we establish that this bound leads to a resolution of the \v{C}ern\'{y} problem for carefully synchronizing automata.
We also study the set of matrices with no zero rows and columns, denoted by $\mathscr{NZ}$, due to its intriguing connections to the \v{C}ern\'{y} conjecture and the recent generalization of Perron-Frobenius theory for this class.
We characterize computational complexity of different problems related to the exponent of $\mathscr{NZ}$ matrix sets,
and present a quadratic bound on the exponents of sets belonging to a special subclass. Namely, we show that the exponent of a set of matrices having total support is bounded by $2n^2 -5n +5$.
\end{abstract}

\section{Introduction}
A nonnegative matrix $M$ of size $n\times n$ is called \emph{primitive} if $M^{k}$ is positive (i.e. has all its entries larger than zero) for a positive integer $k$. This notion was introduced by Frobenius in 1912 during the development of so-called \emph{Perron-Frobenius theory}. This theory has found numerous applications since then: in the theory of Markov chains, economics, population modelling, centrality measures in networks, see~\cite[Chapter 8]{Meyer2000} for an introduction to the topic.
Motivated by various applications Protasov and Voynov introduced the following generalization of this notion to sets of matrices~\cite{PrVo2012}: 
a finite set of (entrywise) nonnegative matrices $\mathcal{M}=\{M_1, M_2, \ldots, M_k\}$ is called \emph{primitive} if $M_{i_1} M_{i_2} \ldots M_{i_m}$ is (entrywise) positive for some indices $i_1, i_2, \ldots, i_m \in [1, k]$. The length of the shortest such product is called the exponent $\exp{(\mathcal{M})}$ of $\mathcal{M}$. We will denote the value of the largest exponent among all sets of $n \times n$ matrices by $\exp(n)$. For example, the matrix set $\mathcal{M}$ in Fig~\ref{fig:primitive} is primitive, since the product $M_1M_2M_1M_2$ is entrywise positive, and its exponent is equal to $4$. Since the actual values of positive entries of matrices in $\mathcal{M}$ do not influence the exponent, in the rest of our paper we will implicitly assume that entries of all matrices are equal to 0 or 1. Moreover, we assume that the product $AB$ of two matrices of size $n \times n$ is also a $(0,1)$-matrix that is defined as follows\footnote{Formally speaking, we consider the matrices over the Boolean semiring.}: $AB[i,j] = 1$ if $\sum_{k}A[i,k]B[k,j]>0$, and $AB[i,j]=0$ otherwise.

Primitive sets of matrices received a lot of attention for different reasons. We refer the reader to the introduction in~\cite{BJO15} for the account of applications of primitive matrix sets to stochastic control theory and to the consensus problem. The connections to contractive matrix families and scrambling matrices are given in detail in~\cite[Section 5]{PrVo2012}. Primitive sets of matrices further arise in the study of time-inhomogeneous Markov chains~\cite{Hart02}, and are of importance in mathematical ecology~\cite{Log10}.
Furthermore, primitive sets of matrices are tightly related to boolean networks, which are widely used in biology to model gene regulatory networks. A special class of boolean networks -- disjunctive networks, can be seen as a set matrices over the Boolean semiring.
While researchers in theoretical biology are mostly interested in the attractors and the limit cycles for different types of update schedules, see for example~\cite{GoNo12}, we are mainly interested whether it is possible and how fast one can achieve the all-one state. 
The subfamily of nonnegative matrices that have no zero rows and columns, denoted by $\mathscr{NZ}$, will be of major interest to us for the following reasons. A matrix $M$ is called \emph{irreducible} if for every $i,j$ there exists a positive integer $m$ such that $M^{m}[i,j]>0$. A set of $k \geq 1$ matrices $\mathcal{M}$ is \emph{irreducible} if the matrix $\sum_{i=1}^{k} M_i$ is irreducible. As usual, we will denote by $e_i$ the $i$th vector of the \emph{canonical basis} in $\mathbb{R}^n$, the $i$th entry of $e_i$ is one, all the others are zeros. We say that a matrix $M$ acts as a permutation on a partition $V_1, V_2, \ldots, V_m$ of the vectors of the canonical basis if there exists a permutation $\sigma$ such that for all $i$, $V_i M$ belongs to the subspace spanned by $V_{\sigma(i)}$. A classical theorem of Perron-Frobenius theory states that an irreducible matrix $M$ is primitive if and only if there is no partition $V_1, \ldots, V_m$ of the canonical basis vectors for $m>1$ such that $M$ acts as a permutation on $V_1, \ldots, V_m$. Protasov and Voynov generalized this theorem to sets of matrices belonging to $\mathscr{NZ}$~\cite{PrVo2012}: an irreducible set of matrices belonging to $\mathscr{NZ}$ is primitive if and only if there is no partition $V_1, \ldots, V_m$ for $m>1$ such that every $M \in \mathcal{M}$ acts as a permutation on $V_1, \ldots, V_m$. Thus, the class of primitive matrices belonging to $\mathscr{NZ}$ can be viewed as the right class for Perron-Frobenius-type theory of matrix sets. This characterization also leads to an efficient algorithm that decides whether a set of matrices belonging to $\mathscr{NZ}$ is primitive.

\subsection{Synchronizing automata}
A deterministic finite state automaton $\mathrsfs{A}$ is a triple\footnote{The classical definition also involves an initial and a set of final states. Since they don't play any role in our considerations, we will omit them.} $\langle Q, \Sigma, \delta \rangle$, where $Q$ is a finite set of states, $\Sigma$ is a finite set of input symbols called \emph{the alphabet}, and $\delta$ is a transition function $\delta : Q \times \Sigma \rightarrow Q$. The image of a state $q$ under the action of a word $w$ is denoted by $q \cdot w$. An automaton $\mathrsfs{A}$ is called \emph{synchronizing} if there exist a word $w$ and a state $f$ such that for every state $q$ we have $q\cdot w = f$. Any such word is called a \emph{synchronizing} or \emph{reset} word. The length of the shortest such word is called \emph{the reset threshold} $\rt (\mathrsfs{A})$ of $\mathrsfs{A}$. Synchronizing automata naturally appear in different areas of research. For example, they were used to model sensorless parts orienting problems: given a part and a set of available actions that can change its spatial orientation, find a sequence of actions that would bring the part to a desired orientation independently of the initial position~\cite{Na1986}. Clearly, if we consider an automaton $\mathrsfs{A}$ with the set of spatial orientations as the set of states, and the available actions as letters, then the ``orienting sequence'' corresponds to a synchronizing word of $\mathrsfs{A}$. We refer the reader to~\cite{Volkov2008Survey} for the survey of main results and other applications. A recent account of applications of synchronizing automata in group theory can be found in~\cite{ArCaSt15}. Persisting interest of the research community to the topic is also driven by one of the most famous open problems in automata theory. Namely, \emph{the \v{C}ern\'{y} conjecture} states that the reset threshold of an $n$-state automaton is at most $(n-1)^2$~\cite{Cerny1964,CernyPirickaRosenauerova1971}. This bound is reached by the $n$-state \v{C}ern\'{y} automaton $\mathrsfs{C}_n$, see~\cite[p. 18]{Volkov2008Survey}, but despite intensive efforts of researchers, the best upper bound $\tfrac{n^3-n}{6}$ was obtained more than 30 years ago in~\cite{Pin1983OnTwoCombinatorialProblems, Fr1982} and independently in~\cite{KRS1987}.

The notion of a synchronizing automaton can be generalized in three different ways to nondeterministic automata~\cite{ImSt99}. We will focus our attention on the most relevant for us. An automaton $\mathrsfs{A}$ is a \emph{partial} automaton if the transition function $\delta$ is partial, i.e. there might be undefined transitions for some pairs of states and letters. A partial automaton is \emph{carefully synchronizing} if there exist a word $w$ and a state $f$ such that $q\cdot w$ is defined and equal to $f$ for every state $q$. Any such word is called a \emph{carefully synchronizing} word. The length of the shortest such word we will denote by $\car (\mathrsfs{A})$.
We will denote by $\car(n)$ the maximum of $\car(\mathrsfs{A})$ among all $n$-state partial automata. Essentially, carefully synchronizing automata model the problem of bringing a simple finite-state device to a known state with a single input sequence, while avoiding undefined transitions, which are undesirable or can break the device. In matrix terms, it amounts to consider a set of matrices with \emph{at most} one 1-entry per row, and to ask for a product with one (entrywise) positive column.

\subsection{Our contributions}
Our results can be informally arranged into three different groups. The contributions of the first group significantly improve the understanding of the relationships between primitive sets of matrices and synchronizing automata. The work within this framework started in~\cite{AGV2013}, where well-known examples of primitive matrices with large exponent were used to construct series of automata with relatively large reset thresholds, so-called ``slowly synchronizing automata''. In~\cite{BJO15} it was shown that a $f(n)$ bound on the reset threshold of $n$-state automata implies a $2f(n)+n-1$ bound on the exponent of $\mathscr{NZ}$ matrix sets. We significantly improve these results. We show that the growth rate of $\exp(n)$ is equal to $\Theta(\car(n))$. Thus, in a certain sense, the study of the exponents of sets of matrices is equivalent to the study of carefully synchronizing automata. We also formulate an analogous result for primitive $\mathscr{NZ}$ matrix sets. Namely, we introduce a special class of automata $\mathscr{C}$ such that the growth rate of the reset thresholds of automata in this class is equivalent to the growth rate of the exponents of $\mathscr{NZ}$ matrix sets. We propose and formalize a new open question of whether a quadratic bound on $\exp_{\mathscr{NZ}}(n)$ leads to a breakthrough on the \v{C}ern\'{y} conjecture.

The contributions of the second group are of combinatorial nature. Our main result states that $\lim_{n\rightarrow\infty} \tfrac{\log \exp(n)}{n} = \tfrac{\log 3}{3}$, and equivalently, $\lim_{n\rightarrow\infty} \tfrac{\log \car(n)}{n} = \tfrac{\log 3}{3}$. From the automata theory point of view our contribution can be seen as the resolution of the \v{C}ern\'{y}-like problem for the carefully synchronizing automata. From the point view of matrix theory, our result is a generalization of the classical theorem by Wielandt that the exponent of a single matrix is at most $(n-1)^2 + 1$, see for example~\cite[Corollary 8.5.9]{HoJo}. It also answers the question of establishing the growth rate of $\exp(n)$ posed in~\cite{BJO15}. Another contribution in this group is a partial result for $\mathscr{NZ}$ matrix sets. Recall that a matrix $M$ has \emph{total support} if every non-zero element $m_{i,j}$ of $M$ lies on a positive diagonal, i.e. for every $i,j \in [1,n]$ such that $m_{i,j}>0$ there exists a permutation $\sigma$ with the following properties: $\sigma(i)=j$ and for every $k \in [1,n]$ we have $m_{k, \sigma (k)} > 0$. We prove that the exponent of a set of matrices having total support is bounded by $2n^2 -5n +5$. In the proof we  utilize the well-known theorem by Kari that the reset threshold of an Eulerian automaton is bounded by $n^2 -3n +3$. This result suggests that the bounds for other classes of synchronizing automata might be used to obtain upper bounds on the exponent in the special classes of $\mathscr{NZ}$ matrix sets.

The contributions of the last group are related to the computational complexity of finding the exponent of an $\mathscr{NZ}$ matrix set. Given a set of two matrices belonging to $\mathscr{NZ}$ and possibly an integer $k$ encoded in binary, we establish the exact computational complexity 	of the following problems:
\begin{enumerate}
\item the problem of deciding whether $\exp(\mathcal{M}) \leq k$ is $NP$-complete;
\item the problem of deciding whether $\exp(\mathcal{M}) = k$ is $DP$-complete;
\item the problem of computing $\exp(\mathcal{M})$ is $FP^{NP[\log]}$-complete.
\end{enumerate} 
Furthermore, we show that unless $P=NP$, for every positive $\varepsilon$ there is no polynomial-time algorithm that computes the exponent of an $\mathscr{NZ}$ matrix set with the approximation ratio $n^{1-\varepsilon}$, even in the case of only three matrices in the set. These results are based on a single relatively simple reduction from automata with a sink state to sets of matrices belonging to $\mathscr{NZ}$.

The paper is organized as follows. Section 2 deals with the primitive sets of matrices in the general case. We show that $\exp(n)=\Theta(\car(n))$ and prove that $\lim_{n\rightarrow\infty} \tfrac{\log \exp(n)}{n} = \tfrac{\log 3}{3}$. Section 3 is devoted to the $\mathscr{NZ}$ matrix sets. In subsection 3.1 we introduce the class $\mathscr{C}$ such that $\exp_{\mathscr{NZ}}(n) = \Theta (\rt_{\mathscr{C}}(n))$. We also present a quadratic bound on the exponent of a set of matrices having total support. In subsection 3.2 we deal with the complexity issues related to the computation of the exponent of $\mathscr{NZ}$ matrix sets.

\section{The general case}

Recall that we denote the value of the largest exponent among all $n \times n$ matrices by $\exp(n)$. The growth rate of $\exp (n)$ is one of the most basic questions one can ask about the sets of primitive matrices. Furthermore, an upper bound on $\exp(n)$ gives a bound on the running time of the straightforward algorithm that decides whether a given set of matrices is primitive: we iterate through all the possible products of length up to $\exp(n)$ and check, whether  they contain a positive matrix. Since the problem is $NP$-hard~\cite[Theorem 6]{BJO15}, such a simple algorithm might be the best we can hope for. The best known bounds on $\exp (n)$ were presented in~\cite[Theorem 10]{BJO15}:
\begin{theorem}
\label{th:expOld}
If $\mathcal{M}$ consists of $m$ matrices of size $n \times n$ then $\exp(\mathcal{M}) \leq 2^{n^2}$. Moreover, if $m \geq 4$, then for all $\varepsilon > 0$ there exists a sequence of positive integers $n_1, n_2, \ldots , $ tending to infinity such that $((1-\varepsilon)e)^{\sqrt{n_k/2}} \leq \exp (n_k)$.
\end{theorem}

Recall that we denote the maximum of $\car(\mathrsfs{A})$ among all $n$-state automata $\mathrsfs{A}$ by $\car(n)$.
In the upcoming theorem we are going to show that $\exp (n)$ grows asymptotically as $\car(n)$. Thus, we can utilize the known bounds on $\car(n)$ to infer the bounds on $\exp (n)$. Furthermore, in the next subsection we will be able to significantly improve the known upper bound on $\car(n)$, and equivalently, on $\exp(n)$. Before stating the theorem we require one last definition. Given a (partial or complete) automaton $\mathrsfs{A}$, an \emph{adjacency matrix} $M_\ell$ of a letter $\ell$ is defined as follows: $M_{\ell}[i,j]=1$ if $i \cdot \ell = j$, and $M_\ell [i,j]=0$ otherwise. In Fig.~\ref{fig:primitive} the matrix $M_1$ is an adjacency matrix of the letter $m_1$.

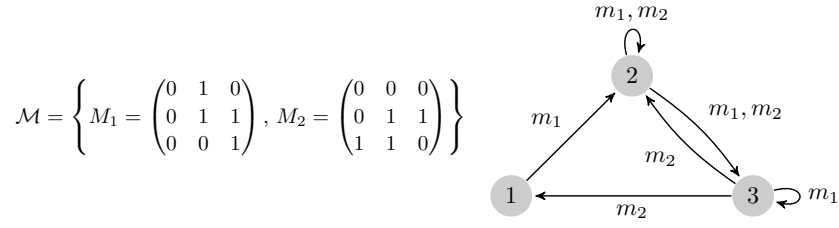
\begin{figure}
\begin{center}
\scalebox{0.8}{
$\mathcal{M} = \left\{ M_1=
\begin{pmatrix}
0 & 1 & 0\\
0 & 1 & 1\\
0 & 0 & 1
\end{pmatrix},\, M_2 = 
\begin{pmatrix}
0 & 0 & 0\\
0 & 1 & 1\\
1 & 1 & 0
\end{pmatrix}
\right\}
$
}
\scalebox{0.9}{
\begin{tikzpicture}[->,>=stealth',shorten >=1pt,auto,node distance=2.5cm,
                    semithick, baseline={([yshift=-1ex]current bounding box.center)}]
  \tikzstyle{every state}=[fill=light-gray,draw=none,text=black, scale=1,minimum size=0.5cm]

  \node[state] 		   (A) {$2$};
  \node[state]         (B) [below right of=A] {$3$};
  \node[state]         (C) [below left of=A] {$1$};

  \path (A) edge [bend left=10] node {$m_1,m_2$} (B)
			edge [loop above]  	node {$m_1,m_2$} (A)
        (B) edge [bend left=10]	node {$m_2$} (A)
            edge node {$m_2$} (C)
            edge [loop right]  	node {$m_1$} (B)
        (C) edge  node {$m_1$} (A);
\end{tikzpicture}}
\end{center}
\caption{The matrix set $\mathcal{M}$ and the corresponding non-deterministic automaton $\mathrsfs{M}$.}
\label{fig:primitive}
\end{figure}

\begin{theorem}
\label{th:expcar}
Let $\exp (n)$ be the maximum value of the exponent among all sets of $n \times n$ matrices. Let $\car(n)$ be the maximum value of $\car(\mathrsfs{A})$ among all $n$-state partial automata $\mathrsfs{A}$, then $\exp(n) = \Theta(\car(n))$.
\end{theorem}
\begin{proof}
The proof of the first part of the theorem is inspired by~\cite[Theorem 16]{BJO15}. 
While, in~\cite{BJO15}, the result was restricted to $\mathscr{NZ}$ matrices, we extend it here to all primitive sets of matrices.  
Furthermore, we make it deterministic, which will be crucial for Theorem~\ref{th:nzrcclass}.
Let us consider an arbitrary primitive set of matrices $\mathcal{M} = \{ M_1, \ldots, M_k \}$. We are going to show now that $\exp(\mathcal{M}) \leq 2\car(n) + n-1$, which implies $\exp(n)=O(\car(n))$. 
We will achieve this by presenting products $P,Q,R$ of matrices in $\mathcal{M}$ with the following properties: 
\begin{enumerate}
\item the $i$th column of $P$ is positive for some $i$ and the length of $P$ is at most $\car(n)$;
\item the $j$th row of $R$ is positive for some $j$ and the length of $R$ is at most $\car(n)$;
\item $Q[i,j]>0$ and the length of $Q$ is at most $n-1$.
\end{enumerate}
These properties clearly imply that $PQR$ is positive and the length of $PQR$ is at most $2\car(n) + n-1$. Thus, $\exp(\mathcal{M}) \leq 2\car(n) + n-1$.

We will construct the product $P$ by utilizing a partial automaton $\mathrsfs{A}$ defined as follows: a partial function $\varphi: [1,n] \rightarrow [1,n]$ is a letter of $\mathrsfs{A}$ if and only if there is a matrix $M \in \mathcal{M}$ with the properties: for all $i$, if $\varphi(i)$ is defined, $M[i, \varphi(i)]>0$, otherwise the $i$th row has no positive entries. First, we are going to show that $\mathrsfs{A}$ is carefully synchronizing, then we will use the shortest carefully synchronizing word of $\mathrsfs{A}$ to obtain the matrix product $P$.

We construct a carefully synchronizing word of $\mathrsfs{A}$ with the help of an auxiliary non-deterministic automaton $\mathrsfs{M}$ defined in the following manner:  the set of states of $\mathrsfs{M}$ is equal to $[1,n]$; for each matrix $M_i \in \mathcal{M}$ we add a letter $m_i$ such that $M_i$ is the adjacency matrix of $m_i$, see Fig.~\ref{fig:primitive}. It is straightforward to verify that for every $s,t,p_1, \ldots, p_\ell$ we have $M_{p_1}M_{p_2}\ldots M_{p_\ell}[s,t] >0$ if and only if there is a path from $s$ to $t$ in $\mathrsfs{M}$ labelled by the word $m_{p_1}m_{p_2}\ldots m_{p_\ell}$. Since $\mathcal{M}$ is primitive, there exists a positive product of matrices in $\mathcal{M}$. Therefore, there exists a word $w$ such that for every pair of states of $\mathrsfs{M}$ there is a path between them labelled by $w$. 

It remains to show that the word $w$ can be transformed to a carefully synchronizing word of $\mathrsfs{A}$. Let us fix a state $t \in [1,n]$. There are paths $\pi_1, \ldots, \pi_n$ in $\mathrsfs{M}$ labelled by $w$ and for every $s$ the path $\pi_s$ goes from the state $s$ to the state $t$. Furthermore, we can impose an additional property on these paths. Namely, if at a step $h$ paths $\pi_x$ and $\pi_y$ are in the same state, then their continuations coincide. Indeed, let $\pi_x = \pi'_x v \pi''_x$ and $\pi_y = \pi'_y v \pi''_y$. Then we can substitute the path $\pi_y$ with the path $\pi'_y v \pi''_x$, which still goes from $y$ to $t$ and it is labelled by $w$.
Observe now, that the paths $\pi_1, \ldots, \pi_n$ can be easily treated as paths leading to the state $t$ in the partial automaton $\mathrsfs{A}$: by construction for each letter $m$ of $\mathrsfs{M}$ and states $v_{i_1}, v_{i_2}, \ldots v_{i_p}$ with a property $v_{i_x} \in \delta_{\mathrsfs{M}}(i_x, m)$ for $x \in [1,p]$, there exists a letter $\ell$ of $\mathrsfs{A}$ such that $\delta_{\mathrsfs{A}}(i_x, \ell) = v_{i_x}$ for each $x \in [1,p]$; due to this fact and the unique continuation property of the paths, we conclude that there exists a word $w'$ over the alphabet of $\mathrsfs{A}$ that labels the paths from every state to the state $t$ in $\mathrsfs{A}$. Thus, $\mathrsfs{A}$ is carefully synchronizing.

Let $w=w_1w_2 \ldots w_h$ be the shortest carefully synchronizing word of $\mathrsfs{A}$. It is easy to see that a product $W_1W_2 \ldots W_h$ contains a column of ones, where $W_x$ is the adjacency matrix of $w_x$ for $x \in [1,h]$. Since for every $x\in [1,h]$ there is matrix $A_x \in \mathcal{M}$ such that $W_x \leq A_x$ we obtain a product $P=A_1A_2 \ldots A_f$ with the properties: $P$ has a column of ones and its length is bounded by $\car (n)$.

The product $R$ is constructed in the same manner by applying the reasoning of the previous paragraphs to a matrix set $\mathcal{M}^T = \{ M^T \mid M \in \mathcal{M}\}$. The resulting product $R^T$ has a column of ones and the length at most $\car (n)$. The existence of the product $Q$ easily follows from the fact that $\mathrsfs{M}$ is strongly connected (otherwise the set of matrices $\mathcal{M}$ is not primitive). Thus, for every pair of states $i,j$ there exists a path of length at most $n-1$ that bring $i$ to $j$.

Now, given a carefully synchronizing $n$-state automaton $\mathrsfs{A}$ with the reset threshold equal to $\car(n)$ we will construct a primitive set of matrices $\mathcal{M}$ such that $\car(\mathrsfs{A}) \leq \exp(\mathcal{M})$. It will imply $\exp(n) = \Omega (\car (n))$.
Let $e$ be a row vector of 1's, and $e_k$ be a row vector with the only non-zero entry equal to 1 at position $k$. Let $\mathcal{E} = \{e_k^T e \mid k \in [1,n]\}$. The set of matrices $\mathcal{M}$ is defined as a union $\mathcal{M'}\cup\mathcal{E}$, where $\mathcal{M'}$ is a set of the adjacency matrices of letters of the partial automaton $\mathrsfs{A}$. Since $\mathrsfs{A}$ is carefully synchronizing, there is a product $P$ of matrices in $\mathcal{M'}$ such that the $i$th column is positive for some $i$. If we multiply $P$ by the matrix $e_k^T e \in \mathcal{E}$ on the right, we obtain a positive matrix product. Thus, $\mathcal{M}$ is primitive.

It remains to show that $\car(\mathrsfs{A}) \leq \exp(\mathcal{M})$.
Let $W$ be the shortest positive product of matrices in $\mathcal{M}$. Note, that $W$ contains at least one matrix from $\mathcal{E}$, since every product of matrices in $\mathcal{M'}$ contains at most one 1 in each row. Let $W = UEV$, where the product $U$ doesn't contain matrices from $\mathcal{E}$ and $E \in \mathcal{E}$. Observe that $U$ contains a positive column. Otherwise, $W$ will have a zero row due to the presence of a zero row in $UE$. Therefore, the length of the product $U$ is at least $\car(\mathrsfs{A})$ and we obtain the desired inequality.
\end{proof}

\begin{corollary}
\label{corr:simpleBounds}
The growth rate of $\exp(n)$ is $O(n^2 4^{\frac{n}{3}})$ and $\Omega(3^{\frac{n}{3}})$.
\end{corollary}
\begin{proof}
The first part of the claim follows from the result of Zs. Gazdag et al.~\cite[Theorem 3]{Ivan2009}: 
$\car(n) = O(n^2 4^{\frac{n}{3}})$. Thus, $\exp(n) = \Theta(\car(n)) = O(n^2 4^{\frac{n}{3}})$.
The second part follows from the result of Martyugin~\cite{Mart2010}. He constructed a series of carefully synchronizing automata with the length of the shortest carefully synchronizing word equal to $\Omega(3^{\frac{n}{3}})$. Thus, $\exp(n) = \Omega(3^{\frac{n}{3}})$.
\end{proof}

\subsection{Improving the upper bound on the exponent}
\label{sec:imprupper}

The goal of this section is to significantly improve the bound on $\exp(n)$ and, equivalently, on $\car(n)$. 
We will present a new upper bound on the length of the shortest carefully synchronizing word by modifying constructions from~\cite{Ivan2009}.

Recall that a \emph{partition} of a set $Q$ is a collection $\{Q_1, Q_2, \ldots, Q_k\}$ of pairwise disjoint non-empty sets whose union is equal to $Q$. Given a partition $\mathscr{P} =\{Q_1, Q_2, \ldots, Q_k\}$ of $Q$, a set $S \subseteq Q$ is called a {\it transversal} of $Q$ with respect to the partition $\mathscr{P}$ if for each $Q_i \in \mathscr{P}$ there is a unique $s \in S$ such that $s \in Q_i$.
A set $S \subseteq Q$ is a {\it partial transversal} with respect to $\mathscr{P}$ if for each $Q_i \in \mathscr{P}$ there is at most one $s \in S$ such that $s \in Q_i$.

\begin{example}
For a partition $\{\{1,2\}, \{3,4\}, \{5\} \}$ of $\{1,2,3,4,5\}$ the sets $\{1,3,5\}$ and $\{1,4,5\}$ are transversals and $\{1,3\}$ and $\{5\}$ are partial transversals. The set $\{1,2\}$ is neither transversal, nor partial transversal.
\end{example}
Let $Q$ be an $n$-element set and $\mathscr{P}$ be an arbitrary partition. We will denote by $\mathscr{T}(\mathscr{P})$ the number of different transversals with respect to $\mathscr{P}$ and by $\mathscr{T}^\ell (\mathscr{P})$ the number of different partial transversals of size $\ell$.
Let $\mathscr{T}_k(n)$ be the largest value of $\mathscr{T}(\mathscr{P})$ among all partitions $\mathscr{P}$ of $Q$ into $k$ parts. Similarly, let $\mathscr{T}^{\ell}_k(n)$ be the largest value of $\mathscr{T}^\ell(\mathscr{P})$ among all partitions $\mathscr{P}$ of $Q$ into $k$ parts. If the value of $n$ is clear from the context, then we will often write $\mathscr{T}_k$ and $\mathscr{T}^{\ell}_k$ to simplify notation.
We will make use of the following bounds on $\mathscr{T}_k(n)$ and $\mathscr{T}^{\ell}_k(n)$:

\begin{lemma}
\label{lemma:trans}
\begin{enumerate}
\item $\mathscr{T}_k(n) \leq 2^{n-k}$ for $k \in [1,n]$.
\item $\mathscr{T}_k(n) \leq 2^{3k-n}3^{n-2k}$ for $\frac{n}{3} \leq k \leq \frac{n}{2}$.
\item $\mathscr{T}_k(n) \leq 3^{\frac{n}{3}}$ for $k \leq \frac{n}{3}$.
\item $\mathscr{T}_k^{k-j}(n) \leq {n\choose j} \mathscr{T}_k (n)$ for $j \in [0,k-1]$.
\end{enumerate}
\end{lemma}
\begin{proof}
\begin{enumerate}
\item It is the statement of Proposition 5 in~\cite{Ivan2009}.
\item Let $\mathscr{P}$ be a partition of $Q$ into $k$ parts such that $\mathscr{T} (\mathscr{P}) = \mathscr{T}_k(n)$, where $\tfrac{n}{3} \leq k \leq \tfrac{n}{2}$. If $d_i$ is the size of the $i$th part of $\mathscr{P}$, then it is easy to see that $\mathscr{T}_k(n) = \prod_{i=1}^{k}d_i$. Observe that for any $i,j$ we have $d_j - d_i < 2$. Otherwise, by moving an element from the $j$th part to the $i$th part of the partition $\mathscr{P}$, we will increase the number of transversals: $(d_i + 1)(d_j - 1) = d_id_j + (d_j - d_i -1) > d_id_j$. Therefore, for a given range of values $k$, every $d_i$ is equal to $2$ or $3$. Let $x$ be the number of $d_i$'s equal to $2$, then $k-x$ is the number of $d_i$'s equal to $3$. Since $2x + 3(k - x) = n$, we derive that $x = 3k-n$, and the desired bound follows.
\item Let $\mathscr{P}$ be a partition of $Q$ into $k$ parts such that $\mathscr{T} (\mathscr{P}) = \mathscr{T}_k(n)$, where $ k \leq \tfrac{n}{3}$. If $d_i$ is the size of the $i$th part of $\mathscr{P}$, then by the inequality of arithmetic and geometric means we have $$\mathscr{T}_k(n) = \prod_{i=1}^{k}d_i \leq \left(\frac{\sum_{i=1}^{k} d_i}{k}\right)^k=\left(\frac{n}{k}\right)^k.$$ Let us bound the right hand side. Note that $\tfrac{\partial}{\partial k}(\tfrac{n}{k})^k = (\ln(\tfrac{n}{k}) - 1)(\tfrac{n}{k})^k$. For $k \leq \tfrac{n}{3}$, we have $\ln(\tfrac{n}{k}) - 1 > 0$. Therefore, the largest value of the function $(\tfrac{n}{k})^k$ is achieved at $k=\tfrac{n}{3}$.
Thus, $\mathscr{T}_k(n) \leq (\tfrac{n}{k})^k \leq 3^{\frac{n}{3}}$.
\item Let $\mathscr{P}$ be a partition of $Q$ into $k$ parts such that $\mathscr{T}^{k-j} (\mathscr{P}) = \mathscr{T}^{k-j}_k(n)$ and let $d_i$ be the size of the $i$th part of $\mathscr{P}$.

$$\mathscr{T}_k^{k-j}(n) = \sum_{|I|=k-j} \prod_{i \in I} d_i \leq {k \choose k-j} \prod_{i=1}^{k}d_i = {k \choose j} \prod_{i=1}^{k}d_i \leq {n \choose j} \mathscr{T}_k(n).$$
\end{enumerate}
\end{proof}

\begin{theorem}
\label{th:newBound}
Let $\exp (n)$ be the maximum value of the exponent among all sets of $n \times n$ matrices. Let $\car(n)$ be the maximum value of $\car(\mathrsfs{A})$ among all $n$-state partial automata $\mathrsfs{A}$, then $\lim_{n\rightarrow\infty} \tfrac{\log \exp(n)}{n} = \tfrac{\log 3}{3}$, and equivalently $\lim_{n\rightarrow\infty} \tfrac{\log \car(n)}{n} = \tfrac{\log 3}{3}$.
\end{theorem}
\begin{proof}
We will show that $\car(n)$ is at most $(3+ \varepsilon)^{n/3}$ for any $\varepsilon >0$ once $n>n(\varepsilon)$ for some threshold $n(\varepsilon)$. Since $\car(n)$ is $\Omega(3^{\frac{n}{3}})$ by~\cite{Mart2010}, the statement $\lim_{n\rightarrow\infty} \tfrac{\log \car(n)}{n} = \tfrac{\log 3}{3}$ will clearly follow. Due to Theorem~\ref{th:expcar} we will have the same statement for $\exp(n)$.

Let $\mathrsfs{A}$ be a carefully synchronizing $n$-state partial automaton with the set of states $Q$.
We will construct a carefully synchronizing word $u_1$ of $\mathrsfs{A}$ via the following iterative procedure:
\begin{enumerate}[(a)]
\item Let $u_{n-1}$ be a letter that is defined on every state $q \in Q$ and satisfies $|Q \cdot u_{n-1}| < |Q|$, where $|\cdot|$ denotes the cardinality of a set. Since $\mathrsfs{A}$ is carefully synchronizing, there exists at least one such letter.
\item Choose a positive integer $\ell < k$. Let $u_{k-\ell}$ be a word of the form 
$$ u_{k-\ell} = u_k t^k_{k} u_k t^{k-1}_k u_k \ldots t^{k - \ell+1}_k u_k, $$
where the words $t_k^{k-s+1}$ are defined iteratively for $s\in [1,\ell]$: $t_k^{k-s+1}$ is the shortest word such that $u_{k}t_k^k \ldots t_k^{k-s+1}u_k$ is defined on every state and $|Q \cdot u_{k}t_k^k \ldots t_k^{k-s+1}u_k| \leq k - s$. 
\end{enumerate}
Note, that the word $u_{k-\ell}$ is well-defined, since at every step the set of possible words for $t_k^{k-s+1}$ contains carefully synchronizing words of $\mathrsfs{A}$. Our procedure further ensures that $|Q \cdot u_k|\leq k$ for every $k$. Thus, $u_1$ is indeed a carefully synchronizing word. Our goal now is to bound the length of $u_1$.
The bound on $\car(n)$ presented in~\cite{Ivan2009} was obtained using the presented procedure with the parameter $\ell$ ultimately fixed to 1. By choosing for every $\varepsilon>0$ a sufficiently large $\ell$ satisfying certain conditions, we get a significant improvement. We proceed by bounding the length of intermediate words $u_k$. To simplify the presentation and without loss of generality, we will further assume that $n$ is divisible by $6\ell$.
\begin{enumerate}
\item $\frac{n}{2} \leq k \leq n-1$. For these values of $k$ we uniformly put $\ell=1$. The proof of this case is presented in~\cite[Proposition 7]{Ivan2009}, which we repeat here for convenience.
We are going to show that $|u_k| \leq (n - k) 2^{n-k-1}$ by induction. Note, that $|u_{n-1}| = 1$. Let $u_{k-1} = u_k t_k^k u_{k}$.
The word $u_k$ gives rise to a partition $\mathscr{P}$ of $Q$ into $k$ parts as follows: a pair of states $p,q$ belongs to the same part of $\mathscr{P}$ if $p\cdot u_k = q \cdot u_k$. Observe that if $t_k^k=x_1x_2\ldots x_m$, where $x_1,\ldots,x_m$ are letters, then the set $Q \cdot u_k x_1\ldots x_i$ is a transversal with respect to $\mathscr{P}$ for every $i \in [1,m-1]$. Indeed, if it is not the case for some $i'$, then the word $u_k x_1 \ldots x_{i'}u_k$ satisfy the conditions $(b)$ of our procedure and it is shorter than $u_{k-1}$, which is impossible. Therefore, the length of $t_k^k$ is bounded by $\mathscr{T}(\mathscr{P})$. By lemma~\ref{lemma:trans} we conclude $|t_k^k| \leq \mathscr{T}_k(n) \leq 2^{n-k}$. Therefore, $|u_{k-1}| \leq 2 |u_k| + |t_k^k| \leq 2 (n - k) 2^{n-k-1} + 2^{n-k} = (n - k + 1)2^{n-k}$, which completes the induction. Observe, that for $k = \tfrac{n}{2}$ we have $|u_{\frac{n}{2}}| \leq \tfrac{n}{2} 2^{\frac{n}{2}}$.

\item $\frac{n}{3} \leq k < \frac{n}{2}$. For each $\varepsilon>0$ we will choose the value of $\ell$ at the end of the proof, independent of $n$ and $k$. As before, the word $u_k$ gives rise to a partition $\mathscr{P}$ of $Q$ into $k$ parts as follows: a pair of states $p,q$ belongs to the same part of $\mathscr{P}$ if $p\cdot u_k = q \cdot u_k$. 
Let us fix $s \in [1,\ell]$ and let $t_k^{k-s+1} = x_1x_2\ldots x_m$, where $x_1,\ldots,x_m$ are letters. By construction, for every $i \in [1,m-1]$ the cardinality of the set $S_i = Q \cdot u_k t_k^k u_k \ldots t_k^{k-s} u_k x_1 \ldots x_i$ is equal to $k-s+1$. Furthermore, $S_i$ is a partial transversal with respect to $\mathscr{P}$. 
Indeed, if it is not the case, then $|S_{i}\cdot u_k| \leq k-s$ and the length of $t_k^{k-s+1}$ is not minimal. 
Therefore, the length of $t_k^{k-s+1}$ is bounded by the number of partial transversals $\mathscr{T}^{k-s+1}(\mathscr{P}) \leq \mathscr{T}_k^{k-s+1}(n)$. 
Using this bound for all $s \in [1,\ell]$ and part 4 of lemma~\ref{lemma:trans} we obtain
$$|u_{k-\ell}| \leq (\ell+1) |u_k| + \left({n\choose 0} + {n \choose 1} + \ldots + {n \choose \ell-1} \right)\mathscr{T}_k = (\ell+1) |u_k| +  p(n) \mathscr{T}_k,$$
for some polynomial $p(n)$ of degree $\ell-1$.
Note, that part 3 of lemma~\ref{lemma:trans} can be rewritten as $\mathscr{T}_{\frac{n}{2} - j} \leq 2^{\frac{n}{2}}\left( \tfrac{9}{8} \right)^j$ for $j \in [0,\tfrac{n}{6}]$. Applying this inequality and the inequality for $|u_{k-\ell}|$ $\frac{n}{6 \ell}$ times we derive:
\begin{align*}
|u_\frac{n}{3}| &\leq |u_\frac{n}{2}| (\ell + 1)^\frac{n}{6\ell} + \mathscr{T}_\frac{n}{2} p(n) (\ell+1)^{\frac{n}{6\ell}-1} + \mathscr{T}_{\frac{n}{2} - \ell} p(n) (\ell+1)^{\frac{n}{6\ell}-2} \\
 &\quad + \mathscr{T}_{\frac{n}{2} - 2\ell} p(n) (\ell+1)^{\frac{n}{6\ell}-3} + \ldots + \mathscr{T}_{\frac{n}{3} + 2\ell} p(n)(\ell + 1) + \mathscr{T}_{\frac{n}{3} + \ell} p(n)\\
&\leq \frac{n}{2} 2^\frac{n}{2} (\ell + 1)^\frac{n}{6\ell} + 2^\frac{n}{2} p(n) (\ell+1)^{\frac{n}{6\ell}-1}+ 2^\frac{n}{2} \left(\frac{9}{8}\right)^{\ell}  p(n) (\ell+1)^{\frac{n}{6\ell}-2} \\
&\quad + 2^\frac{n}{2} \biggl(\frac{9}{8}\biggr)^{2\ell} \!\! p(n) (\ell+1)^{\frac{n}{6\ell}-3} + \ldots + 2^\frac{n}{2} \biggl(\frac{9}{8}\biggr)^{\frac{n}{6} - 2\ell} \!\!\!\!  p(n)(\ell+1) + 2^\frac{n}{2} \biggl(\frac{9}{8}\biggr)^{\frac{n}{6} - \ell} \!\! p(n).\\
\end{align*}
By choosing $\ell$ large enough, such that it satisfies $\frac{(\frac{9}{8})^\ell}{\ell+1} > 1$, we ensure that 
every term, starting from the second one, is majorated by the last term. Thus, $$ |u_\frac{n}{3}| \leq \frac{n}{2} 2^\frac{n}{2} (\ell + 1)^\frac{n}{6\ell} + \frac{n}{6\ell} 2^\frac{n}{2} \left(\frac{9}{8}\right)^{\frac{n}{6} - \ell} p(n) \leq 2^\frac{n}{2} \left(\frac{9}{8}\right)^{\frac{n}{6}}  \left( \frac{n}{2} + \frac{n}{6\ell} p(n) \right)  \leq 3^{\frac{n}{3}} q(n),$$
for another polynomial $q(n)$.

\item $1 \leq k < \frac{n}{3}$. 
Since $\lim_{\ell\to\infty} (\ell+1)^{1/\ell} = 1$, for every $\varepsilon>0$ we can choose $\ell$ such that $(\ell+1)^{1/\ell} < 1 + \frac{\varepsilon}{3}$. In the same way as before, we derive
\begin{align*}
|u_1| &\leq |u_\frac{n}{3}| (\ell + 1)^\frac{n}{3\ell} + \mathscr{T}_\frac{n}{3} p(n) (\ell+1)^{\frac{n}{3\ell}-1} + \mathscr{T}_{\frac{n}{3} - \ell} p(n) (\ell+1)^{\frac{n}{3\ell}-2}\\
&\quad+ \ldots + \mathscr{T}_1 p(n) \\ 
&\leq 3^\frac{n}{3} q(n) (\ell + 1)^\frac{n}{3\ell} + 3^\frac{n}{3} p(n) (\ell+1)^{\frac{n}{3\ell}-1} + 3^\frac{n}{3} p(n) (\ell+1)^{\frac{n}{3\ell} - 2}\\
&\quad+ \ldots + 3^\frac{n}{3} p(n) \leq 3^\frac{n}{3} (\ell+1)^{\frac{n}{3\ell}} r(n) < 3^{\frac{n}{3}}\left(1 + \frac{\varepsilon}{3}\right)^{\frac{n}{3}} = 
(3+\varepsilon)^{\frac{n}{3}},
\end{align*}
where $r(n)$ is a polynomial. The last inequality holds for large enough $n$.
\end{enumerate}
\end{proof}

\section{Sets with no zero rows nor zero columns}

\subsection{Bounds on the exponent}
A quadratic lower bound on the exponents of sets of matrices belonging to $\mathscr{NZ}$ was obtained in~\cite[Corollary 20]{BJO15}.
A first cubic upper bound $\tfrac{n^3+n^2-4n+2}{2}$ was given in~\cite[Theorem 1]{Voy13}\footnote{At the discussion of the connections with the \v{C}ern\'{y} conjecture the author refers to a wrong bound, see~\cite{GoJuTr15} for a discussion.}. The proof relies on standard linear algebraic techniques. This bound was improved in~\cite[Corollary 18]{BJO15} to $\tfrac{n^3+2n-3}{3}$. The proof is based on the following fact:
a bound $f(n)$ for the reset thresholds of synchronizing automata implies a bound $O(f(n))$ for the exponents of $\mathscr{NZ}$ matrix sets~\cite[Theorem 17]{BJO15}. In this subsection we will extend this result and present a quadratic bound for a special class of $\mathscr{NZ}$ matrix sets.

We denote by $\exp_{\mathscr{NZ}}(n)$ the maximal exponent among all primitive matrix sets belonging to $\mathscr{NZ}$.
In order to state a theorem for $\exp_{\mathscr{NZ}}(n)$, analogous to theorem~\ref{th:expcar}, we will introduce a new class of automata $\mathscr{C}$ defined as follows. An automaton $\mathrsfs{A}$ with the set of states $[1,n]$ over an alphabet $\Sigma$ belongs to $\mathscr{C}$ if there exists a partition of $\Sigma$ into $\Sigma_1, \Sigma_2, \ldots, \Sigma_k$ such that for every $i \in [1,k]$ we have:
\begin{enumerate}
\item for each state $q$ there exists a state $p$ and a letter $\ell \in \Sigma_i$ such that $p \cdot \ell = q$;
\item for every choice of states $q_1, \ldots , q_n$ such that $j \cdot \ell_j = q_j$ for some $\ell_j \in \Sigma_i$, there exists a letter $\ell \in \Sigma_i$ with the property $j \cdot \ell = q_j$ for all $j \in [1,n]$.
\end{enumerate}
In other words, for each $i$, every state is reachable from somewhere by a letter in $\Sigma_i$, and given a list of transformations of states by letters in $\Sigma_i$, we can find a letter in $\Sigma_i$ that performs all the transformations at once. 
\begin{example}
The \v{C}ern\'{y} automaton $\mathrsfs{C}_n$ equipped with an identity letter $c$ belongs to $\mathscr{C}$. The required partition is $\Sigma_1 = \{a,c\}$ and $\Sigma_2 = \{b\}$. Clearly, the reset threshold is not changed with the addition of the letter $c$.
\end{example}
\begin{theorem}
\label{th:nzrcclass}
Let $\exp_{\mathscr{NZ}}(n)$  be the maximum value of the exponent among all sets of $\mathscr{NZ}$ matrices of size $n \times n$. Let $\rt_{\mathscr{C}}(n)$ be the largest reset threshold among $n$-state automata in $\mathscr{C}$, then $\exp_{\mathscr{NZ}}(n) = \Theta(\rt_{\mathscr{C}}(n))$.
\end{theorem}
\begin{proof}
In order to show that $\exp_{\mathscr{NZ}}(n)$ is $O(\rt_{\mathscr{C}}(n))$ we will reuse the reduction from primitive sets of matrices to partial automata presented in the first part of theorem~\ref{th:expcar}. Let $\mathcal{M}$ be a primitive set of $\mathscr{NZ}$ matrices. It can be reduced to partial automata $\mathrsfs{A}$ and $\mathrsfs{B}$ such that $\exp(\mathcal{M}) \leq \car(\mathrsfs{A}) + \car(\mathrsfs{B}) + n-1$. Since $\mathcal{M}$ is an $\mathscr{NZ}$ matrix set, we conclude that the automata $\mathrsfs{A}$ and $\mathrsfs{B}$ are complete. Thus, their carefully synchronizing words are ordinary synchronizing words. Furthermore, the automata $\mathrsfs{A}$ and $\mathrsfs{B}$ belong to $\mathscr{C}$. Indeed, every letter of these automata was obtained from a matrix in $\mathcal{M}$ or $\mathcal{M}^T$. In order to obtain the desired partition, we group a pair of letters together if and only if they were derived from the same matrix. It is a straightforward check that both conditions on the partition are satisfied.

For the other direction, given an automaton $\mathrsfs{A} \in \mathscr{C}$ we will construct a set of $\mathscr{NZ}$ matrices $\mathcal{M}$ such that $\rt(\mathrsfs{A}) \leq \exp(\mathcal{M})$. It will imply that $\exp_{\mathscr{NZ}}(n)=\Omega(\rt_{\mathscr{C}}(n))$. Let $\Sigma_1, \ldots, \Sigma_k$ be the partition of the letters of $\mathrsfs{A}$. A set of matrices $\mathcal{M}$ consists of matrices $M_i = \sum_{\ell \in \Sigma_i} A_\ell$, where $A_\ell$ is the adjacency matrix of the letter $\ell$. Clearly, each $M_i$ is an $\mathscr{NZ}$ matrix due to the first property of the partition and the fact that $\mathrsfs{A}$ is complete. The second property ensures that $\rt(\mathrsfs{A}) \leq \exp(\mathcal{M})$, since every primitive word of $\mathcal{M}$ can be transformed into a synchronizing word of $\mathrsfs{A}$ as in the proof of theorem~\ref{th:expcar}.
\end{proof}

\begin{problem}
\label{prob:rtnzrc}
Improve the bounds $O(n^3)$ and $\Omega(n^2)$ on the growth rate of $\rt_{\mathscr{C}}(n)$ and, equivalently, $\exp_{\mathscr{NZ}}(n)$. In particular, is there a constant $K$ such that $\rt_{\mathscr{C}}(n) < K n^2$?
\end{problem}
This problem can be settled in different ways. On one hand, one can show that problem~\ref{prob:rtnzrc} is as hard as the \v{C}ern\'{y} conjecture. There are not many natural problems equivalent to it and the problem of bounding $\exp_{\mathscr{NZ}}(n)$ is a good candidate for this purpose. On the other hand, a quadratic bound for $\exp_{\mathscr{NZ}}(n)$ is clearly of interest by itself.

In the remainder of this subsection we will present an upper bound on the exponent of a set of matrices from a special class. A matrix $M$ has \emph{total support} if every non-zero element $m_{i,j}$ of $M$ lies on a positive diagonal, i.e. for every $i,j \in [1,n]$ such that $m_{i,j}>0$ there exists a permutation $\sigma$ with the following properties: $\sigma(i)=j$ and for every $k \in [1,n]$ we have $m_{k, \sigma (k)} > 0$. The class of matrices with total support received a lot of attention in the past. For example, it appears in the necessary and sufficient condition for the convergence of the classical Sinkhorn-Knopp method for matrix scaling, see~\cite{SK67}. Another characterization is related to a class of doubly stochastic matrices.
A square matrix $M$ is called \emph{doubly stochastic} if the entries are nonnegative and the sum of elements in each row and column is equal to 1.
A matrix $M$ is said to have a \emph{doubly stochastic pattern} if there exists a doubly stochastic matrix $D$ such that for all $i,j$ it holds $D[i,j]>0$ if and only if $M[i,j]>0$. A famous result of Perfect and Mirsky~\cite{PM65} states that a matrix $M$ has total support if and only if $M$ has a doubly stochastic pattern, see~\cite[Theorem 9.2.1]{Bru06} for a modern and much more general treatment of the problem. Now we are ready to state our result.

\begin{theorem}
If each matrix of a primitive set $\mathcal{M}$ has total support, then the exponent of $\mathcal{M}$ is at most $2n^2 -5n + 5$, where $n \times n$ is the size of matrices in $\mathcal{M}$.
\end{theorem}
\begin{proof}
We will modify the reduction presented in the first part of theorem~\ref{th:expcar} from a primitive set of matrices $\mathcal{M}$ to partial automata $\mathrsfs{A},\mathrsfs{B}$ in order to prove the statement. By the aforementioned result of Perfect and Mirsky we conclude that for every matrix $M \in \mathcal{M}$ there exists a doubly stochastic matrix $D_1$ such that for every $i,j$ we have $M[i,j]>0$ if and only if $D_1[i,j]>0$. It is not hard to see that we can further assume that $D_1$ has only rational entries. Therefore, there exists $h$ and a matrix $D_M$ with nonnegative integer elements such that the sum of entries in each row and column is equal to $h$, and $M[i,j]>0$ if and only if $D_M[i,j]>0$ for all $i,j$.  We will define the automaton $\mathrsfs{A}$ as follows. For each matrix $M \in \mathcal{M}$ we add a function $\varphi : [1,n] \rightarrow [1,n]$ as a letter of $\mathrsfs{A}$ multiple times. Namely, we treat $\varphi(\cdot)$ as $\prod_{i=1}^{n} D_M[i, \varphi(i)]$ different letters of the automaton $\mathrsfs{A}$. Let $\mathrsfs{B}$ be an automaton obtained from a matrix set $\mathcal{M}^T$ in the same manner. Similarly to the proofs of theorems~\ref{th:expcar} and~\ref{th:nzrcclass} we can conclude that the automata $\mathrsfs{A}$ and $\mathrsfs{B}$ are complete and synchronizing, moreover, $\exp(\mathcal{M}) \leq \rt(\mathrsfs{A}) + \rt(\mathrsfs{B}) + n - 1$.

Now we are going to show that the automaton $\mathrsfs{A}$ is Eulerian, i.e. the in-degree of each state is equal to its out-degree. Let $\Sigma_M$ be the set of letters generated from a matrix $M \in \mathcal{M}$ and $h$ be the row (and column) sum of $D_M$. It is not hard to see that by the definition of $\mathrsfs{A}$ the size of $\Sigma_M$ is $h^n$. Furthermore, the number of letters from $\Sigma_M$ that move a state $i$ to a state $j$ is equal to $D_M[i,j]h^{n-1}$. Thus, the number of incoming edges to $j$ labelled by a letter from $\Sigma_M$ is equal to $\sum_{i}D_M[i,j]h^{n-1}=h^n$, which is equal to the size of $\Sigma_M$. Since the alphabet of $\mathrsfs{A}$ is $\cup_{M \in \mathcal{M}} \Sigma_M$ and for every $\Sigma_M$ the number of incoming and outgoing edges labelled by $\Sigma_M$ coincide, we conclude that the automaton $\mathrsfs{A}$ is Eulerian. The same reasoning allow us to conclude that the automaton $\mathrsfs{B}$ is also Eulerian.
By the famous result of Kari about the reset thresholds of Eulerian automata~\cite{Kari2003Eulerian} we have $\rt(\mathrsfs{A}),\rt(\mathrsfs{B}) \leq n^2 - 3n + 3$. Thus, the exponent of the matrix set $\mathcal{M}$ is at most $2 (n^2 -3n + 3) + n-1 = 2n^2 -5n + 5$.
\end{proof}

\subsection{Computation and approximation of the exponent}
In this subsection we will focus on the problem of computing the exponent of a set of $\mathscr{NZ}$ matrices. Our results rely on the following lemma:
\begin{lemma}
\label{lemma:zeronzrc}
For every synchronizing automaton $\mathrsfs{A}$ with a sink state there exists a set of $\mathscr{NZ}$ matrices $\mathcal{M}$ constructible in polynomial time such that $\exp(\mathcal{M}) = \rt(\mathrsfs{A}) + 1$.
\end{lemma}
\begin{proof}
If the number of states of $\mathrsfs{A}$ is equal to 1, then $\mathcal{M}$ can be an arbitrary set of matrices with the exponent equal to 2. For example,
$\mathcal{M} = \{\left( \begin{smallmatrix} 1&1\\ 1&0 \end{smallmatrix} \right)\}$.
In all the other cases we construct the matrix set $\mathcal{M}$ as follows. Since $\mathrsfs{A}$ is synchronizing, it has a unique sink state, which we denote by $s$. Let $\mathcal{M'}$ be the set of adjacency matrices of letters of $\mathrsfs{A}$. The matrix set $\mathcal{M}$ consists of matrices in $\mathcal{M'}$ modified in such a way that the $s$th row of every matrix is positive, i.e. $\mathcal{M} = \{ M' + e_s^Te \mid M' \in \mathcal{M'}\}$, where $e$ stands for a row vector of 1's, and $e_s$ stands for a row vector with the only non-zero entry equal to 1 at position $s$. Since the automaton $\mathrsfs{A}$ is complete, we conclude that every matrix $M \in \mathcal{M}$ has no zero rows. Furthermore, each column of $M$ has a positive element at position $s$. Thus, the set of matrices $\mathcal{M}$ belongs to $\mathscr{NZ}$. Now we will demonstrate that the set $\mathcal{M}$ is primitive. Since the automaton $\mathrsfs{A}$ is synchronizing there exists a product $P'$ of matrices from $\mathcal{M'}$ with a positive column. Due to the fact that for every letter $\ell$ of $\mathrsfs{A}$ one has $s\cdot\ell=s$, we conclude that the only positive entry in the $s$th row of $P'$ is located at position $s$. Thus, the $s$th column of $P'$ is positive. Altering each matrix $M'$ of the product $P'$ to a matrix $M \in \mathcal{M}$ with the property $M' \leq M$ we obtain a product $P$ of matrices from $\mathcal{M}$ with the positive $s$th column. 
Now multiply by any matrix $M \in \mathcal{M}$ to get a positive product $PM$. Thus, the set $\mathcal{M}$ is primitive and $\exp(\mathcal{M}) \leq \rt(\mathrsfs{A}) + 1$.

It remains to show that $\exp(\mathcal{M}) \geq \rt(\mathrsfs{A}) + 1$. Let $P=M_1M_2 \ldots M_m$ be the shortest positive product of matrices in $\mathcal{M}$. Let $P' = M'_1M'_2 \ldots M'_m$ be the corresponding product of matrices in $\mathcal{M'}$, where $M_i$ and $M'_i$ differ only in the $s$th row. For each row $i \neq s$ there exists the largest $h_i$ such that the $i$th row of $P_{h_i}=M_1M_2 \ldots M_{h_i}$ contains a unique positive entry. 
Due to the fact that the $i$th row of $P_{h_i+1}$ has several positive entries and the structure of matrices in $\mathcal{M}$ we conclude that the unique positive entry occupies the position $s$ in the $i$th row of $P_{h_i}$. 
Therefore, the only positive entry in the $i$th row of $P'_{h_i}$ is located at position $s$. Since for every row the value of $h_i$ is strictly less than $m$, we conclude that the product $P'_{m-1}$ has a positive column $s$. It implies that $\rt(\mathrsfs{A}) \leq m-1$. Rearranging, we obtain the desired inequality $\exp(\mathcal{M}) \geq \rt(\mathrsfs{A}) + 1$.
\end{proof}

The computational complexities of problems related to computation of reset thresholds and synchronizing words were extensively studied. Now we will leverage lemma~\ref{lemma:zeronzrc} to easily obtain a large body of results on the computation and approximation of the exponents of $\mathscr{NZ}$ matrix sets. We assume that the reader is familiar with computational complexity theory. All the missing definitions can be found in the book by Arora and Barak~\cite{Aro09}. Recall that $DP$ stands for a class of all languages of the form $L=L_1 \setminus L_2$ with $L_1, L_2 \in NP$. Note, that $DP$ is a superclass of both $NP$ and $coNP$. Thus, the hardness result that we are going to obtain applies for both classes. The class $DP$ is contained in $P^{NP[log]}$ -- the class of problems solvable by a deterministic polynomial-time Turing machine that can use logarithmic number of queries to an oracle for an $NP$-complete problem. It is generally believed that the inclusion is proper. We will denote by $FP^{NP[log]}$ the functional analogue of $P^{NP[log]}$. For a function $f(n) : \mathbb{N} \rightarrow \mathbb{R}$ we say that an algorithm approximates the exponent 
within a factor $f(n)$ if for every set $\mathcal{M}$ of matrices of size $n \times n$ the value $V$ returned by the algorithm satisfies $\exp(\mathcal{M}) \leq V \leq f(n) \exp(\mathcal{M})$.

\begin{theorem}
\label{th:nzrcComplexity}
Given a set $\mathcal{M}$ of three $n \times n$ matrices belonging to $\mathscr{NZ}$ and possibly a positive integer $k$ encoded in binary.
\begin{enumerate}
\item The problem of deciding whether $\exp(\mathcal{M}) \leq k$ is $NP$-complete.
\item The problem of deciding whether $\exp(\mathcal{M}) = k$ is $DP$-complete.
\item The problem of computing $\exp(\mathcal{M})$ is $FP^{NP[\log]}$-complete.
\item For every constant $\varepsilon>0$ it is $NP$-hard to approximate $\exp(\mathcal{M})$ within a factor $n^{1-\varepsilon}$.
\end{enumerate}
\end{theorem}
\begin{proof}
The hardness results for these problems follow from the corresponding statements about the reset thresholds of synchronizing automata with a sink state and lemma~\ref{lemma:zeronzrc}. Eppstein~\cite{Ep1990} proved that it is $NP$-hard to decide whether $\rt(\mathrsfs{A}) \leq k$, where $k$ is given in binary. Olschewski and Ummels~\cite[Theorem 1]{OM2010} shown that it is $DP$-hard to decide whether $\rt(\mathrsfs{A}) = k$. The same authors~\cite[Theorem 4]{OM2010} also proved that the problem of computing $\rt(\mathrsfs{A})$ is $FP^{NP[\log]}$-hard. A recent breakthrough by Gawrychowski and Straszak~\cite[Theorem 16]{GawrychowskiStraszak2015StrongInapproximabilityOfTheShortestResetWord} states that for every constant $\varepsilon>0$ it is $NP$-hard to approximate $\rt(\mathrsfs{A})$ within a factor $n^{1-\varepsilon}$. Moreover, only automata with a sink state were utilized in all reductions presented by the aforementioned authors. In the proof presented by Gawrychowski and Straszak~\cite[Theorem 16]{GawrychowskiStraszak2015StrongInapproximabilityOfTheShortestResetWord} the number of letters $\mathrsfs{A}$ is equal to three, while in all the other reductions automata with only two-letters were utilized.

It remains to show that the stated problems belong to the corresponding classes:
\begin{enumerate}
\item Since $\exp(\mathcal{M})$ is bounded by a low-degree polynomial $p(n)$, e.g. $p(n)=\frac{n^3 +2n -3}{3}$ by~\cite[Corollary 18]{BJO15}, it suffices to guess a positive product of matrices from $\mathcal{M}$ of length $\min\{ p(n), k\}$. Verification for every such product can be done in polynomial time. Thus, the problem belongs to $NP$.
\item It is easy to see that $\exp{\mathcal{M}}=k$ if and only if $\exp{\mathcal{M}} \leq k$ and $\exp{\mathcal{M}} \leq k-1$ does not hold. Since the problem of deciding whether $\exp(\mathcal{M}) \leq k$ belongs to $NP$, we conclude that the problem of deciding whether $\exp(\mathcal{M}) = k$ belongs to $DP$.
\item The algorithm from $FP^{NP[\log]}$ that computes $\exp(\mathcal{M})$ is a simple binary search algorithm that uses logarithmic number of queries to the oracle for the $NP$-complete problem $\exp(\mathcal{M}) \leq k$. Recall that the exponent of $\mathcal{M}$ is bounded by a polynomial $p(n)=\frac{n^3 +2n -3}{3}$. Thus, in $\log p(n)$ iterations we can establish precise value of $\exp(\mathcal{M})$.
\end{enumerate}
\end{proof}

\section{Conclusion}
The goal of our work was to emphasize, and leverage, the fact that several problems about primitive sets of matrices are in some sense equivalent to problems about synchronizing automata. More precisely, we related the bounds on the exponents and the lengths of the shortest carefully synchronizing words in the general case, and the exponents and reset thresholds in the case of matrices without zero rows and columns. Furthermore, we utilized these connections to easily establish the exact complexity classes of different problems concerning the computation of the exponent of a set of matrices belonging to $\mathscr{NZ}$. Thus, we believe that the joint research effort on both topics at the same time can lead to substantial progress on some of the most desperate problems in both fields at the same time. We left a quadratic upper bound on the exponent of an $\mathscr{NZ}$ matrix set as an open problem, and whether its existence brings any implication for the \v{C}ern\'{y} conjecture. Our future work includes the search for matrix counterparts of special classes of automata, which have quadratically bounded reset threshold. These statements will be analogous to our result on primitive sets of matrices having total support.

\subparagraph*{Acknowledgements}
This work was supported by Interuniversity Attraction Poles (IAP) Programme, and by the ARC grant 13/18-054 (Communaut\'e fran\c{c}aise de Belgique). B.~Gerencs\'er was also supported by the Hungarian Academy of Sciences. V.~Gusev benefited from the Russian foundation for basic research (grant 16-01-00795), Ministry of Education and Science of the Russian Federation (project no. 1.1999.2014/K), and the Competitiveness Program of Ural Federal University. R.~Jungers is a F.R.S./F.N.R.S. research associate.

\bibliography{synchronization}

\end{document}